\documentclass[a4paper,12pt]{article}

\usepackage{graphicx}
\usepackage{amsmath,amsthm,amssymb}

\newcounter{thm}
\newcounter{ex}
\newcounter{re}
\newtheorem{Theorem}[thm]{Theorem}
\newtheorem{Lemma}[thm]{Lemma}

\newtheorem{Example}[ex]{Example}
\newtheorem{Remark}[re]{Remark}
\newtheorem{Corollary}[thm]{Corollary}

\def\hf{\frac{1}{2}}
\def\th{\frac{3}{2}}
\def\g{{\mathfrak g}_{\ell}}
\def\h{{\mathfrak h}_{\ell}}
\def\del#1{\partial_{#1}}
\def\hP#1{\hat{P}^{(#1)}}
\def\tP#1{\tilde{P}^{(#1)}}
\def\hD{\hat{D}}
\def\hC{\hat{C}}
\def\tC{\tilde{C}}
\def\hM{\hat{M}}
\def\I#1{{\cal I}_{#1}}

\begin{document}


\begin{center}

{\Large Centrally Extended Conformal Galilei Algebras  and Invariant Nonlinear PDEs}\\
~~\\

\bigskip
{\large Naruhiko Aizawa and Tadanori Kato}\\
~~\\

Department of Mathematics and Information Sciences, \\
Graduate School of Science, 
Osaka Prefecture University, \\
Nakamozu Campus, Sakai, Osaka 599-8531, Japan 

\end{center}

\bigskip
\begin{abstract}
We construct, for any given $ \ell = \hf + {\mathbb N}_0, $ the second-order \textit{nonlinear} partial differential equations 
(PDEs) which are invariant under the transformations generated by the centrally extended conformal Galilei algebras. 
The generators are obtained by a coset construction and the PDEs are constructed by standard Lie symmetry technique. 
It is observed that the invariant PDEs have significant difference for $ \ell > \th. $ 
\end{abstract}



\section{Introduction}

 The purpose of the present work is to construct  partial differential equations (PDEs) which are invariant under the 
transformations generated by the conformal Galilei algebra (CGA). 
We consider a particular realization, which is given in \cite{Aizawa:2013vma}, 
of CGAs with the central extension for the parameters $ (d, \ell)  = (1, \hf + {\mathbb N}_0),  $  
where $ {\mathbb N}_0 $ denotes the set of non-negative integers. 
We also restrict ourselves to the second-order PDEs for computational simplicity. 
Our main focus is on nonlinear PDEs since linear ones have already been discussed in the literatures 
\cite{aizawa2002intertwining,aizawa2008intertwining,Aizawa:2013vma,Aizawa:2014hva}. 
CGA is a Lie algebra which generates conformal transformations in $d+1$ dimensional \textit{nonrelativistic} spacetime 
\cite{negro1997nonrelativistic,negro1997nonrelativistic2,havas1978conformal,henkel1997local}. 
Even in the fixed dimension of spacetime one has infinite number of inequivalent conformal algebras. 
For a fixed value of $d$ each inequivalent CGA is labelled by a parameter $ \ell $ taking the spin value, i.e., 
$ \ell = \hf, 1, \th, 2, \dots. $ Each CGA has an Abelian ideal (namely, CGA is a non-semisimple Lie algebra) so that it would be deformed. 
Indeed, it has a central extension depending the value of the parameters. 
More precisely, there exist two different types of central extensions. 
One of them exists for any values of $d$ and half-integer $ \ell, $ 
another type of extension exists for $ d = 2 $ and integer $ \ell. $ Simple explanation of this fact is found in \cite{martelli2010comments}.

  It has been observed that CGAs for $ \ell = \hf $ and $ \ell = 1 $  play  important roles in various kind of problems in physics and mathematics. 
The simplest $ \ell = \hf $ member of  CGAs is called the Schr\"odinger algebra which was originally discussed by Sophus Lie and Jacobi 
in 19th century \cite{lie1970theorie, jacobi1866vorlesungen} and 
reintroduced later by many physicists \cite{Niederer:1972zz,Niederer:1973tz,Niederer:1974ba,hagen1972scale,jackiw2008introducing,burdet1972many}. 
Recent renewed interest in CGAs is mainly due to the AdS/CFT correspondence. 
The Schr\"odinger algebra and $ \ell=1 $ member of CGA were used to formulate nonrelativistic analogues of 
AdS/CFT correspondence \cite{son2008toward,balasubramanian2008gravity,martelli2010comments}. 
One may find a nice review of various applications of the Schr\"odinger algebras in \cite{unterberger2011schrodinger} 
and see \cite{aizawa2012highest} for more references on the Schr\"odinger algebras and $\ell=1 $ CGAs.  
Physical applications of $ \ell=2$ CGA is found in \cite{henkel2002phenomenology}. 

 Now one may ask a question whether the CGAs with $ \ell > 1 $ are relevant structures to physical or mathematical problems. 
To answer this question one should find classical or quantum dynamical systems relating to CGAs and develop representation theory of CGAs 
(see \cite{aizawa2011irreducible,aizawa2012highest,lu2014simple} for classification of irreducible modules over $d=1, 2$ CGAs ).  
This is the motivation of the present work. We choose a particular differential realization of CGAs then look for PDEs 
invariant under the transformation generated by the realization. Investigation along this line for the Schr\"odinger algebras 
is found in \cite{fushchich1985galilean,fushchich1989galilean,rideau1993evolution,aizawa2002intertwining,aizawa2008intertwining} and for $\ell=1$ CGAs in 
\cite{cherniha2010exotic} and for related algebraic structure in \cite{fushchych1995galilei,cherniha2004non}. 
For higher values of $ \ell $ use of the representation theory such as Verma modules, singular vectors allows us to  
derive \textit{linear} PDEs invariant under CGAs \cite{Aizawa:2013vma,Aizawa:2014hva}.  
More physical applications of CGAs with higher value $\ell$  are found in the literatures 
\cite{duval2009non,duval2011conformal,gomis2012schrodinger,galajinsky2011remarks,galajinsky2013dynamical,galajinsky2013dynamical2,andrzejewski2014conformal,galajinsky2015dynamical,andrzejewski2013dynamical,andrzejewski2013nonrelativistic,andrzejewski2013dynamical2,andrzejewski2013unitary,andrzejewski2014conformal2,andrzejewski2012nonrelativistic,Aizawa:2014uma}. 

 The paper is organized as follows. In the next section the definition of  CGA for  $ (d, \ell)  = (1, \hf + {\mathbb N}_0) $ 
and its differential realization are given. 
Then symmetry of PDEs under a subset of the generators is considered. It is shown that there is a significant distinction of the form of 
invariant PDEs for $ \ell  > \th. $  Invariant PDEs for $ \ell = \th $ CGA are obtained in \S \ref{Sec:TH} 
For $ \ell \geq \frac{5}{2} $ we first derive PDEs invariant under a subalgebra of the CGA in \S \ref{Sec:h}, 
then derive invariant PDEs under full CGA in \S \ref{Sec:g}. \S \ref{Sec:CR} is devoted to concluding remarks.  


\section{CGAs and preliminary consideration}

The CGA for $ d=1 $ and any half-integer $\ell$ consists of $ sl(2,{\mathbb R}) \simeq so(2,1) = \langle\, H, D, C \, \rangle $ and 
$ \ell+1/2$ copies of the Heisenberg algebra $ \langle\, P^{(n)}, M \,\rangle_{n=1,2,\dots, 2\ell+1}.$ 
Their nonvanishing commutators are given by
\begin{eqnarray}
  & & [D, H] = -2H, \quad [D, C] = 2C, \quad [H, C] = D, 
  \nonumber  \\
  & & [H, P^{(n)}] = (n-1)P^{(n-1)}, \quad [D, P^{(n)}]=2(n-1-\ell) P^{(n)},
  \label{CGADefCom} \\
  & & [C, P^{(n)}]=(n-1-2\ell) P^{(n+1)}, \quad
      [P^{(m)}, P^{(n)}] = - \delta_{m+n,2\ell+2} I_{m-1} M,
  \nonumber
\end{eqnarray}
where the structure constant $I_m$ is taken to be $ I_m= (-1)^{m+\ell+\hf} (2\ell-m)! m! $ and $ M$ is the centre of the algebra. 
We denote this algebra by $\g$. 
The subset $ \langle\, P^{(n)}, M, H \,\rangle_{n=1,2,\dots, 2\ell+1}$ forms a subalgebra of $ \g $ and  
we denote it by $\h. $

  We employ the following realization of $ \g$ on the space of functions of the variables $ t = x_0, x_1, \dots, x_{\ell+\hf} $ and $U$ \cite{Aizawa:2013vma}:
\begin{eqnarray}
  & & M = U \del{U}, \qquad D = 2 t \del{t} + \sum_{k=1}^{\ell+\hf} 2(\ell+1-k) x_k \del{x_k}, 
  \qquad H = \del{t}, 
  \nonumber \\
  & & C = t^2 \del{t} + \sum_{k=1}^{\ell+\hf} 2 (\ell+1-k) t x_k \del{x_k} + \sum_{k=1}^{\ell-\hf} (2\ell+1-k) x_k \del{x_{k+1}}
        - \hf \Big( \big(\ell+\hf \big) ! \Big)^2 x_{\ell+\hf}^2 U \del{U},
  \nonumber \\
  & & P^{(n)} = \sum_{k=1}^n \binom{n-1}{k-1} t^{n-k} \del{x_k}, \quad 
      1 \leq n \leq \ell + \hf,
  \label{CGAgenerators} \\
  & & P^{(n)} = \sum_{k=1}^{\ell+\hf} \binom{n-1}{k-1} t^{n-k} \del{x_k} 
      - \sum_{k=\ell+\frac{3}{2}}^n \binom{n-1}{k-1} I_{k-1} t^{n-k} x_{2\ell+2-k} U \del{U}, 
      \quad \ell+\frac{3}{2} \leq n \leq 2\ell+1 \nonumber
\end{eqnarray}
where $ \binom{n}{k} $ is the binomial coefficient and $ I_{k-1} $ is the structure constant appearing in (\ref{CGADefCom}). 
This is in fact a realization of $ \g $ on the Borel subgroup of the conformal Galilei group generated by  $\g$ 
(we made some slight changes from \cite{Aizawa:2013vma}).   
Let us introduce the sets of indices for later convenience:
\begin{equation}
  {\cal I}_{\mu} = \bigl\{\; \mu, \mu+1, \dots, \ell+\hf \; \bigr\}, \qquad \mu = 0, 1, 2, \dots
  \label{IndSets}
\end{equation}
Now we take $ x_{\mu}, \, \mu \in \I{0} $ as independent variables and $ U$ as dependent variable : $ U = U(x_{\mu}). $ 
Our aim is to find second order PDEs which are invariant under the point  transformations generated by (\ref{CGAgenerators}) for 
$ \ell > 1/2 $ ($\ell = 1/2$ corresponds to Schr\"odinger algebra). 
Such a PDE is denoted by
\begin{equation}
   F(x_{\mu}, U, U_{\mu}, U_{\mu\nu}) = 0, \quad 
   U_{\mu} = \frac{\partial U}{\partial x_{\mu}}, \quad U_{\mu\nu} = \frac{\partial^2 U}{\partial x_{\mu} x_{\nu}}.
   \label{InvPDE1}
\end{equation}
We use the shorthand notation throughout this article. The left hand side of (\ref{InvPDE1}) means that 
$F$ is a function of all independent variables $ x_{\mu}, \, \mu \in \I{0}, $ dependent variable $U$ 
and all first and second order derivatives of $U.$  
As found in the standard textbooks (e.g., \cite{olver2000applications, bluman1989symmetries,stephani1989differential}) 
the symmetry condition is expressed in terms of the prolonged generators:
\begin{equation}
   \hat{X} F = 0, \ (\bmod F = 0), \label{GenSymCond}
\end{equation}
where $ \hat{X} $ is the prolongation of the symmetry generator $X$ up to second order:
\begin{equation}
  \hat{X} = X   + \sum_{\mu=0}^{\ell+\hf} \rho^{\mu} \frac{\partial }{\partial U_{\mu}} 
    + \sum_{\mu \leq \nu} \sigma^{\mu\nu} \frac{\partial}{\partial U_{\mu\nu}},
    \qquad
    X = \sum_{\mu=0}^{\ell+\hf} \xi^{\mu} \frac{\partial}{\partial x_{\mu}} + \eta \frac{\partial }{\partial U}.
    \label{ext-generator}
\end{equation}
The quantities $ \rho^{\mu}, \sigma^{\mu\nu} $ are defined by
\begin{eqnarray}
  \rho^{\mu} &=& \eta_{\mu} + \eta_U U_{\mu} - \sum_{\nu=0}^{\ell+\hf} U_{\nu} (\xi^{\nu}_{\mu} + \xi^{\nu}_U U_{\mu}),
     \label{def-rho} \\
  \sigma^{\mu\nu} &=& \eta_{\mu\nu} + \eta_{\mu U} U_{\nu} + \eta_{\nu U} U_{\mu} + \eta_U U_{\mu\nu} + \eta_{UU} U_{\mu} U_{\nu}
     \nonumber \\
     &-& \sum_{\tau=0}^{\ell+\hf} \xi^{\tau}_{\mu\nu} U_{\tau} 
     - \sum_{\tau=0}^{\ell+\hf} (\xi^{\tau}_{\mu} U_{\nu \tau} + \xi^{\tau}_{\nu} U_{\mu \tau}) 
     - \sum_{\tau=0}^{\ell+\hf} \xi^{\tau}_U (U_{\tau} U_{\mu\nu} + U_{\mu} U_{\nu \tau} + U_{\nu} U_{\mu \tau})
     \nonumber \\
     &-& \sum_{\tau=0}^{\ell+\hf} (\xi^{\tau}_{\mu U} U_{\nu} + \xi^{\tau}_{\nu U} U_{\mu} + \xi^{\tau}_{UU} U_{\mu} U_{\nu}) U_{\tau}. 
     \label{def-sigma}
\end{eqnarray}
In this section we consider the symmetry condition (\ref{GenSymCond}) for $ M, H $ and $ P^{(n)}$ with $ n= 1, 2, \dots, 2\ell+1. $

\begin{Lemma} \label{Lem:HPM}
\begin{enumerate}
 \renewcommand{\labelenumi}{(\roman{enumi})}
\item 
Equation (\ref{InvPDE1}) is invariant under $ H, P^{(1)} $ and $M$ if it has the form
  \begin{equation}
     F\left(x_a, \frac{U_{\mu}}{U}, \frac{U_{\mu\nu}}{U}\right) = 0, \quad a  \in \I{2}, \quad \mu, \nu \in \I{0}
     \label{InvEqForm1}
  \end{equation}
\item 
 For $ \ell > \frac{3}{2},$ a necessary condition for the symmetry of the equation (\ref{InvEqForm1})  under $ P^{(n)} $ with 
$ n \in \I{2} $ is that the function $ F$ is independent of $ U_{0m}, \; m \in \I{3}.$
\end{enumerate}
\end{Lemma}
\begin{proof}[Proof of Lemma \ref{Lem:HPM}]
(i) It is obvious from that the generators $ H $ and $ P^{(1)} $ have no prolongation, 
 while the prolongation of $M$ is given by 
 \begin{equation}
      \hat{M} = U\del{U} + \sum_{\mu=0}^{\ell+\hf} U_{\mu} \del{U_{\mu}} + \sum_{\mu \leq \nu} U_{\mu\nu} \del{U_{\mu\nu}}. 
      \label{M-ext}
 \end{equation}
(ii) 
 The lemma is proved by the formula of the prolongation of $ P^{(n)}. $ 
For $ n \in \I{2}  $ the generator $P^{(n)}$ is given by 
\[
  P^{(n)} = \sum_{k=1}^n \binom{n-1}{k-1} t^{n-k} \del{x_k}
\]
and its prolongation yields 
 \begin{equation}
    \hP{n} = \sum_{a=2}^n \begin{pmatrix} n-1 \\ n-a \end{pmatrix} 
         t^{n-a} 
         \left[  
            \tP{a} - \sum_{k=a+1}^n (a-1) U_{a-1\,k} \del{U_{0k}}
         \right],
    \label{PnExt1}
 \end{equation}
 where
 \begin{equation}
    \tP{n} = \left. \hP{n} \right|_{t=0} =  \del{x_n} - (n-1) 
    \left\{ 
      U_{n-1} \del{U_0} + \bigl( (n-2) U_{n-2} + 2U_{0\,n-1} \bigr)\del{U_{00}} + \sum_{k=1}^n U_{n-1\,k} \del{U_{0k}}
    \right\},
    \label{PnExt2}
 \end{equation}
 and the terms containing $ \del{x_1} $ are omitted. 
 We give the explicit expressions for small values of $n$ which will be helpful to see the structure of equation (\ref{PnExt1}):
 \begin{eqnarray*}
   \hP{2} &=& \tP{2}, \\
   \hP{3} &=& \tP{3} + 2t (\tP{2}-U_{13} \del{U_{03}}), \\
   \hP{4} &=& \tP{4} + 3t (\tP{3} - 2 U_{24} \del{U_{04}}) + 3t^2 (\tP{2} - U_{13} \del{U_{03}} - U_{14} \del{U_{04}}).
 \end{eqnarray*}
 Since $ \tP{n} $ is independent of $t,$ each symmetry condition $ \hP{n} F = 0 $ decouples into some independent equations. 
 For example, $ \hP{4} F = 0 $ decouples into the following equations:
 \[
   \tP{4} F = 0, \qquad (\tP{3} - 2 U_{24} \del{U_{04}})F = 0, \qquad (\tP{2} - U_{13} \del{U_{03}} - U_{14} \del{U_{04}}) F = 0
 \] 
 The condition $ \hP{2} F = 0 $ is equivalent to the condition  $ \tP{2} F = 0. $ 
It follows that the  condition $ \hP{3} F =0 $ 
yields two independent conditions $ \tP{3} F = 0 $ and $ U_{13} \del{U_{03}} F = 0. $ 
The second condition means that $F$ is independent of $ U_{03}. $ 
Repeating this for $ \hP{n} F = 0 $ for $ n = 4, 5, \dots, \ell + \hf $  one may prove the lemma.

   Now we show the formula (\ref{PnExt1}).  
   Set $ \xi^k(n) = \begin{pmatrix} n-1 \\ k-1 \end{pmatrix} t^{n-k} $ then 
   $ \displaystyle P^{(n)} = \sum_{k=2}^n \xi^k(n) \del{x_k} $ (recall that we omit $ \del{x_1} $). 
   By the equations (\ref{ext-generator}) - (\ref{def-sigma}) we have
   \begin{equation}
     \hP{n} = P^{(n)} - \sum_{k=1}^n \bigl\{ \xi^k_0(n) U_k \del{U_0} + ( \xi^k_{00}(n) U_k + 2\xi^k_0(n) U_{0k}) \del{U_{00}} 
       + \sum_{m=1}^n \xi^m_0(n) U_{mk} \del{U_{mk}} \bigr\}.
     \label{PnExt3}
   \end{equation}
   Thus the maximal degree of $t$ in $\hP{n} $ is $n-2.$ The following relation is easily verified:
   \begin{equation}
      \frac{\partial^a}{\partial t^a} \xi^k(n) = 
      \left\{
         \begin{array}{lcl}
            \frac{(n-1)!}{(n-a-1)!} \, \xi^k(n-a), & & (1 \leq k \leq n-a) \\[9pt]
            0, & & (n-a < k)
         \end{array}
      \right.
      \label{DelXi}
   \end{equation}
   Using this one may calculate the higher order derivatives of (\ref{PnExt3}): 
   \[
      \frac{\partial^a}{\partial t^a} \hP{n} = \frac{(n-1)!}{(n-a-1)!} 
      \bigl\{
           P^{(n-a)} - \sum_{k=n-a+1}^n \sum_{m=1}^{n-a} \xi^m_0(n-a) U_{mk} \del{U_{0k}}.
      \bigr\}
   \]
   It follows that
   \begin{eqnarray*}
     \hP{n} &=& \sum_{a=0}^{n-2} \frac{1}{a!} \left( \frac{\partial^a}{\partial t^a} \hP{n} \right)_{t=0} t^a 
     \\
     &=& \sum_{a=0}^{n-2} \begin{pmatrix} n-1 \\a \end{pmatrix} 
     t^a \bigl[ \tP{n-a} - \sum_{k=n-a+1}^n (n-a-1) U_{n-a-1\,k} \del{U_{0k}} \bigr].
   \end{eqnarray*}
   By replacing $ n-a $ with $ a $ we obtain the formula (\ref{PnExt1}). The formula (\ref{PnExt2}) is readily obtained by setting 
   $ t = 0 $ in the equation (\ref{PnExt3}). 
\end{proof}
\begin{Remark}
  By Lemma \ref{Lem:HPM} the symmetry condition for $ M, H, P^{(n)} $ with $ n \in \I{1} $ is summarized as
  \begin{equation}
        \tP{n} F\left( x_a, \frac{U_{\mu}}{U}, \frac{U_{00}}{U}, \frac{U_{01}}{U}, \frac{U_{02}}{U}, \frac{U_{km}}{U},  \right) = 0, \quad 
        a \in {\cal I}_2,  \ \mu \in \I{0}, \ k, m  \in {\cal I}_1
        \label{PCond1}
  \end{equation}
  where $ \tP{n} $ is given by (\ref{PnExt2}). 
\end{Remark}

  The condition (\ref{PCond1}) implies that $ F $ is independent of $ U_{00} $ if $ \ell \geq 7/2, $ 
since $ \tP{n} $ has the term $ U_{0k} \del{U_{00}} $ with $ k \geq 3. $ 
In fact one can make a stronger statement by looking at the symmetry conditions for $ P^{(n)} $ with 
$ \ell+\th \leq n \leq 2\ell+1. $ 

\begin{Lemma} \label{Lem:U00}
  $ F $ given in (\ref{PCond1}) is independent of $ U_{00} $ if $ \ell \geq 5/2.$ 
\end{Lemma}
\begin{proof}[Proof of Lemma \ref{Lem:U00}]
 We calculate the prolongation of $ P^{(n)} $ for  $ \ell+\th \leq n \leq 2\ell+1. $ 
The derivatives $ \del{t}, \del{x_1}, \del{U_{0k}} (k \in \I{3}) $ are ignored in the computation. 
Then
\begin{eqnarray}
  P^{(n)} &=& \sum_{m=2}^{\ell+\hf} \xi^m(n) \del{x_m} + \eta(n) \del{U}, 
  \label{PnExt4} \\
   \xi^m(n) &=& \begin{pmatrix} n-1 \\ m-1 \end{pmatrix} t^{n-m}, \qquad 
       \eta(n) = -\sum_{m=\ell+\th}^n \xi^m(n) I_{m-1} x_{2\ell+2-m} U.
  \label{PnExt5}
\end{eqnarray}
One may calculate derivatives of $ \eta(n)$ easily
\[
   \eta_k(n) = \frac{\partial \eta(n)}{\partial x_k} 
   = 
   \left\{
      \begin{array}{lcl}
          -\xi^{2\ell+2-k}(n) I_{2\ell+1-k} U & & 2\ell+2-n \leq k \leq \ell+\hf \\[10pt]
          0                                   & & k <  2\ell+2-n
      \end{array}
   \right.
\]
First and second order derivatives need some care:
\[
   \eta_1(n) = 
   \left\{
      \begin{array}{lcl}
         -I_{2\ell} U & & n = 2\ell+1 \\[10pt]
         0     & & \text{otherwise}
      \end{array}
   \right.
\]
\[
   \eta_2(n) = 
   \left\{
      \begin{array}{lcl}
         -2\ell t I_{2\ell-1} U &  & n = 2\ell+1 \\[10pt]
         -I_{2\ell-1} U & & n = 2\ell \\[10pt]
         0              & & \text{otherwise}
      \end{array}
   \right.
\]
Then a lengthy but straightforward computation gives the following expression for the  prolongation of $ P^{(n)}$ 
up to second order:
\begin{eqnarray}
  \hP{n} &=& \eta_U \hM + \sum_{k=2}^{\ell+\hf} \xi^k(n) \del{x_k} 
    + \bigl( \eta_0(n) - \sum_{k=1}^{\ell+\hf} \xi^k_0(n)U_k \bigr) \del{U_0}
    + \sum_{k=2\ell+2-n}^{\ell+\hf} \eta_k(n) \del{U_k}
   \nonumber \\
   &+& \bigl(
      \eta_{00}(n) + 2 \eta_{0U}(n) U_0 - \sum_{k=1}^{\ell+\hf} ( \xi^k_{00}(n) U_k + 2\xi^k_{0}(n) U_{0k}) 
   \bigr) \del{U_{00}}
   \nonumber \\
   &+& \sum_{k=1,2} \bigl(  \eta_{0U}(n) U_k - \sum_{m=1}^{\ell+\hf} \xi^m_0(n) U_{km} \bigr) \del{U_{0k}}
   - \delta_{n,2\ell} I_{2\ell-1} U_0 \del{U_{02}}
   \nonumber \\
   &-& \delta_{n,2\ell+1} \bigl( I_{2\ell} U_0 \del{U_{01}} + 2 \ell I_{2\ell-1} (U + t U_0) \del{U_{02}} \bigr)
   \nonumber \\
   &+& \sum_{k=2\ell+2-n}^{\ell+\hf} \sum_{m=1}^{\ell+\hf} \eta_{kU} U_m \del{U_{km}} 
   + \sum_{k=2\ell+2-n}^{\ell+\hf} \eta_{kU} U_k \del{U_{kk}}.
   \label{PnExt6}
\end{eqnarray}
We have already taken into account the invariance under $ \hM $ so that 
the first term of (\ref{PnExt6}) is omitted in the following computations. 
It is an easy exercise to verify that
\[
   \frac{\partial^a}{\partial t^a} \eta(n) = 
   \frac{(n-1)!}{(n-a-1)!} \times 
   \left\{
       \begin{array}{lcl}
          0 & & n-\ell-\hf \leq a \\[10pt]
           \eta(n-a) & & 0 \leq a \leq n - \ell -\th
       \end{array}
   \right.
\]
and
\[
   \frac{\partial^a}{\partial t^a} \sum_{k=1}^{\ell+\hf} \xi^k(n) 
   = 
   \frac{(n-1)!}{(n-a-1)!} \times
   \left\{
      \begin{array}{lcl}
         0 & & n \leq a \\[10pt]
         \displaystyle \sum_{k=1}^{n-a} \xi^k(n-a) & & n-\ell-\hf \leq a \leq n-1 \\[15pt]
         \displaystyle \sum_{k=1}^{\ell+\hf} \xi^k(n-a) & & 0 \leq a \leq n-\ell-\th
      \end{array}
   \right.
\]
It follows that for $ 0 \leq a \leq n-\ell-\th $
\begin{equation}
  \frac{\partial^a}{\partial t^a} \hP{n} = \frac{(n-1)!}{(n-a-1)!} \hP{n-a}. \label{PnDer1}
\end{equation}
For $ n-\ell-\hf \leq a \leq n-2 $ (i.e., $ 2 \leq n-a \leq \ell + \hf $) all the derivatives of $ \eta(n)$ vanishes and (\ref{PnExt3}) is recovered. Therefore for all values of $a$ from $0$ to $ n-2 $ the relation (\ref{PnDer1}) holds true. 
Thus we have
\[
  \hP{n}= \sum_{a=0}^{n-2} \frac{1}{a!} \left( \frac{\partial^a}{\partial t^a} \hP{n} \right)_{t=0} t^a 
  = \sum_{a=0}^{n-2} \begin{pmatrix} n-1 \\ a \end{pmatrix} \tP{n-a} t^a 
\]
where $ \tP{n} = \left. \hP{n} \right|_{t=0}. $ 
This means that the symmetry condition under $ \hP{n}$ is reduced to
\begin{equation}
   \tP{n} F = 0, \qquad \ell + \th \leq n \leq 2\ell+1 \label{PCond2}
\end{equation}

 Now we look at the part containing $ \del{U_{00}} $ in  (\ref{PnExt6}), namely, the second line of the equation.
The contribution to $ \tP{n} $ from the  term $ \displaystyle \sum_k \xi^k_0(n)U_{0k} \del{U_{00}} $ is  $ U_{0\,\ell+\hf} \del{U_{00}}.$ 
Since $ \ell+\hf \geq 3 $ for $  \ell \geq \frac{5}{2} $ the condition (\ref{PCond2}) gives $ \del{U_{00}} F = 0 $ for this range 
of $ \ell.$ 
Thus $ F $ has $ U_{00} $ dependence only for $ \ell=\th.$ 
\end{proof}

 Lemma \ref{Lem:U00} requires a separate treatment of the case  $ \ell = \th. $
In the following sections we solve the symmetry conditions (\ref{PCond1}) 
and (\ref{PCond2}) explicitly for $ \ell = \th $ and for $ \ell > \th $ separately.  
Before proceeding further we here present the formulae of prolongation of $ D$ which is not difficult to verify: 
\begin{eqnarray}
  \hD &=& \sum_{k=2}^{\ell+\hf} 2(\ell+1-k) x_k \del{x_k} -2 U_0 \del{U_0} - \delta_{\ell,\th}\, 4 U_{00} \del{U_{00}}
  - 2\sum_{k=1,2} (\ell+2-k) U_{0k} \del{U_{0k}}
  \nonumber \\
  &-& 2 \sum_{k=1}^{\ell+\hf} \bigl[ (\ell+1-k) U_k \del{U_k} + \sum_{m =k}^{\ell+\hf}  (2\ell+2-k-m) U_{km} \del{U_{km}}  \bigr].
  \label{DExt} 
\end{eqnarray}
The prolongation of $C$ is more involved so we present it in the subsequent sections separately for $ \ell = \th $ and for 
other values of $\ell. $

%
\section{The case of $ \ell = \th$}
\label{Sec:TH}

The goal of this section is to derive the PDEs invariant under the group generated by $ {\mathfrak g}_{\th}. $ 
First we solve the conditions (\ref{PCond1}) and (\ref{PCond2}). We have from (\ref{PnExt2})
\begin{equation}
  \tP{2} = \del{x_2} - \bigl( U_1 \del{U_0} + 2 U_{01} \del{U_{00}} + \sum_{k=1,2} U_{1k} \del{U_{0k}} \bigr) 
  \label{TildP1} 
\end{equation}
and collecting the $ t=0 $ terms of (\ref{PnExt6})
\begin{eqnarray}
  \tP{3} &=& -2 \bigr[ U_2 \del{U_0} + U \del{U_2} + (U_1 + 2 U_{02}) \del{U_{00}} + U_{12} \del{U_{01}} 
         + (U_0 + U_{22}) \del{U_{02}} + U_1 \del{U_{12}} + 2 U_2 \del{U_{22}} \bigr],
   \nonumber \\
  \tP{4} &=& -6 \bigr[ x_2 U \del{U_0} - U \del{U_1} + (U_2 + 2x_2 U_0) \del{U_{00}} + (x_2 U_1 -U_0) \del{U_{01}} 
          +(x_2 U_2 + U) \del{U_{02}}
   \nonumber \\
   & & \qquad 
     - (2U_1 \del{U_{11}} + U_2 \del{U_{12}}) \bigr].
\end{eqnarray}
The symmetry conditions (\ref{PCond1}) and (\ref{PCond2}) are the system of first order PDEs so that it can be solved 
by the standard method of characteristic equation (e.g., \cite{courant1966methods}). 
It is not difficult to verify that the following functions are the solutions to (\ref{PCond1}) and (\ref{PCond2}). 
\begin{eqnarray}
  \phi_1 &=& \frac{U_{11}}{U} - \left( \frac{U_1}{U} \right)^2, \qquad 
  \phi_2 = \frac{U_{22}}{U} - \left( \frac{U_2}{U} \right)^2, 
  \nonumber \\
  \phi_3 &=& \frac{U_{12}}{U} - \frac{U_1 U_2}{U^2}, \qquad \quad 
  \phi_4 = \frac{U_0}{U} + \frac{x_2U_1}{U} - \frac{U_{22}}{2U},
  \nonumber \\  
  \phi_5 &=& \frac{U_{01}}{U} - \frac{U_0 U_1}{U^2} + x_2 \phi_1 - \frac{U_2}{U} \phi_3,
  \nonumber \\
  \phi_6 &=& \frac{U_{02}}{U} + \frac{U_1}{U} - \frac{U_0 U_2}{U^2} - \frac{U_2}{U} \phi_2+ x_2 \phi_3,
  \nonumber \\
  \phi_7 &=& \frac{U_{00}}{U} - \left( \frac{U_0}{U} \right)^2 - \left( \frac{U_1}{U} + \frac{2U_{02}}{U} \right) \frac{U_2}{U} 
     + \left( \frac{2U_0}{U} + \frac{U_{22}}{U} \right) \left( \frac{U_2}{U} \right)^2 
  \nonumber \\
         &-& \left( \frac{U_2}{U} \right)^4 - x_2^2 \phi_1 + 2 x_2 \phi_5. 
   \label{InvFunc1}
\end{eqnarray}
Thus we have proved the following lemma:
\begin{Lemma} \label{Lem:HMPall}
Equation (\ref{InvPDE1}) is invariant under $ {\mathfrak h}_{\th} = \langle \; M, H,   P^{(n)} \; \rangle_{n=1,2,3,4} $ 
if it has the form
 \begin{equation}
    F(\phi_1, \phi_2, \dots, \phi_7) = 0. \label{InvEqForm2}
 \end{equation}
\end{Lemma}

 Next we consider the further invariance under $ D $ and $ C. $ 
The computation of the second order prolongation of $ C  $ for $ \ell = \th $ is straightforward based on (\ref{ext-generator})-(\ref{def-sigma}). It has the form
\begin{eqnarray}
  \hC &=& -2x_2^2\, \hM + t \hD + 3x_1 \tP{2} - \tC,
  \nonumber \\
  \tC &=& x_2 U_2 \del{U_0} + 3 U_2 \del{U_1} + 4 x_2 U \del{U_2} + 2(U_0+x_2 U_{02}) \del{U_{00}} 
     + (3U_1 + 3 U_{02} + x_2 U_{12}) \del{U_{01}}
  \nonumber \\
     &+& (U_2 + 4x_2 U_0 + x_2 U_{22}) \del{U_{02}} + 6 U_{12} \del{U_{11}} + (4x_2 U_1 + 3 U_{22}) \del{U_{12}}
     + 4(U+2x_2 U_2) \del{U_{22}}.
   \label{CExt1}
\end{eqnarray}
It is an easy exercise to see the action of $ \tC $ on $\phi_k:$
\begin{eqnarray*}
  & & \tC \phi_1 = 2\phi_3, \qquad \tC \phi_2 = \frac{4}{3}, \qquad \tC \phi_3 = \phi_2, 
  \qquad \tC \phi_4 = -\frac{2}{3}, \qquad \tC \phi_5 = \phi_6, 
  \nonumber \\
  & & \tC \phi_6 = 0, \qquad\quad  \tC \phi_7 = \frac{1}{3} \phi_2 + \frac{2}{3} \phi_4. 
\end{eqnarray*}
It follows that the following combinations of $ \phi_k $ are invariant of $ \tC:$
\begin{eqnarray}
  & & w_1 = \hf \phi_2 + \phi_4, \qquad w_2 = 2\phi_3 - \frac{3}{4} \phi_2^2, \qquad w_3 = \frac{1}{2\sqrt{2}}\phi_6,
  \qquad w_4 = \phi_1 - \frac{3}{2} \Bigl(w_2 \phi_2 + \frac{1}{8} \phi_2^3 \Bigr),
  \nonumber \\
  & & 
    w_5 = \frac{1}{\sqrt{2}}\phi_5 - \frac{3}{4\sqrt{2}} w_3 \phi_2, \qquad 
    w_6 = \phi_7 - \frac{1}{2} w_1 \phi_2.
    \label{InvFunc2}
\end{eqnarray}
On the other hand $\hD$ generates the scaling of $ w_k: $ 
\begin{equation}
  \begin{array}{clclc}
    w_1 \ \to \  e^{2\epsilon} w_1, & \quad &  w_2 \ \to \  e^{4\epsilon} w_2, & \quad  & w_3 \ \to \  e^{3\epsilon} w_3,
    \\
        w_4 \ \to \  e^{6\epsilon} w_4, & \quad &  w_5 \ \to \  e^{5\epsilon} w_5, & \quad  & w_6 \ \to \  e^{4\epsilon} w_6.
    \end{array}
  \label{Dscale1}
\end{equation}
With these observations one may construct all invariants of the group which is generated by $ {\mathfrak g}_{\th}:$
\begin{eqnarray}
  & & \psi_1 = \frac{w_2}{w_1^2} = \frac{\Psi_1}{\Phi^2}, \qquad 
  \psi_2 = \frac{w_3^2}{w_1^3} = \frac{\Psi_2^2}{\Phi^3}, 
  \qquad 
  \psi_3 = \frac{w_4}{w_3^2} = \frac{\Psi_3}{\Psi_2^2},
  \nonumber \\
  & & \psi_4 =\frac{w_5^2}{w_1^5} = \frac{\Psi_4^2}{\Phi^5},
  \qquad
  \psi_5 = \frac{w_6}{w_2} = \frac{\Psi_5}{\Psi_1},
  \label{InvFunc3}  
\end{eqnarray}
where
\begin{eqnarray}
  \Phi &=& 2(U_0 + x_2 U_1) U - U_2^2,
  \nonumber \\
  \Psi_1 &=&  8 (U_{12}U - U_1 U_2) U^2 - 3 (U_{22}U-U_2^2)^2,
  \nonumber \\
  \Psi_2 &=& (U_1 +  U_{02}) U^2 - U_2 \bigl( (U_0+U_{22}) U - U_2^2 \bigr)  + x_2( U_{12} U - U_1 U_2) U,
  \nonumber \\
  \Psi_3 &=& 8 U_{11} U^5 - 8 U_1^2 U^4 -12(U_{22}U - U_2^2) (U_{12} U -U_1 U_2) U^2 + 3(U_{22}U-U_2^2)^3,
  \nonumber \\
  \Psi_4 &=& 4 U_{01} U^4 - 4 U_0 U_1 U^3 - 3\bigl( (U_1+U_{02})U - U_0 U_2 \bigr) (U_{22}U-U_2^2) U
      + 3U_2 (U_{22}U-U_2^2)^2
    \nonumber \\
    &+& x_2 \bigl( 4U_{11}U^3 - 4 U_1^2 U^2 - 3(U_{22}U-U_2^2) (U_{12}U-U_1 U_2)  \bigr) U,
  \nonumber \\
  \Psi_5 &=& 4 U_{00} U^3 - 2\bigl( 2U_0^2 + U_0 U_{22} + 2(U_1+2U_{02})U_2 \bigr) U^2
      + 5 (2U_0 + U_{22}) U_2^2 U - 5 U_2^4
  \nonumber \\  
     &+& 2x_2 \bigl( 4U_{01}U^2 - (4 U_0 U_1 + 4 U_2 U_{12} + U_1 U_{22})U + 5 U_1 U_2^2 \bigr) U
     + 4x_2^2 (U_{11} U - U_1^2) U^2.
    \label{PsiPhi}
\end{eqnarray}
Thus we obtain the PDEs with the desired symmetry. 
\begin{Theorem} \label{Thm:threehalf}
 The PDE invariant under the Lie group generated by  the realization  (\ref{CGAgenerators}) of  $ {\mathfrak g}_{\th}$ is 
 given by
  \begin{equation}
    F(\psi_1, \psi_2, \psi_3, \psi_4, \psi_5) = 0 \label{InvEqThm1}
  \end{equation}
  where $ F$ is an arbitrary differentiable function and $ \psi_i $ is given in (\ref{InvFunc3}). 
  Explicit form of the symmetry generators are as follows:
  \begin{eqnarray}
    & & M = U \del{U}, \qquad \,D = 2t \del{t} + 3 x_1 \del{x_1} + x_2 \del{x_2}, \qquad 
        H = \del{t},
    \nonumber \\
    & & C = t(t\del{t}+3x_1 \del{x_1} + x_2 \del{x_2}) + 3 x_1 \del{x_2}-2x_2^2 U \del{U},
    \nonumber \\
    & & P^{(1)} = \del{x_1}, \qquad P^{(2)} = t \del{x_1} + \del{x_2},
    \nonumber \\
    & & P^{(3)} = t^2 \del{x_1} + 2t \del{x_2} -2 x_2 U \del{U},
    \nonumber \\
    & & P^{(4)} = t^3 \del{x_1} + 3t^2 \del{x_2} -6(t x_2 -x_1) U \del{U}.
  \end{eqnarray}
\end{Theorem}

%
\section{The case of $ \ell \geq \frac{5}{2} $ : $\h$-symmetry}
\label{Sec:h}

  As shown in Lemma \ref{Lem:U00} the function $F$ is independent of $ U_{00} $ so that the PDE which we have at this stage is of  the form
\begin{equation}
F\left( x_a, \frac{U_{\mu}}{U},  \frac{U_{01}}{U}, \frac{U_{02}}{U}, \frac{U_{km}}{U},  \right) = 0, \quad 
        a \in {\cal I}_2,  \ \mu \in \I{0}, \ k, m  \in {\cal I}_1
        \label{InvEqForm3}
\end{equation}
We wants to make the PDE (\ref{InvEqForm3}) invariant under all the generators of $\h.$ 
Invariance under $M $ and $P^{(1)} $ has been completed. We need to consider the invariance under  $ P^{(n)} $ for $ n \in \I{2}. $ 
The symmetry conditions are (\ref{PCond1}) and (\ref{PCond2}). We give $ \tP{n} $ more explicitly. 
From (\ref{PnExt2}) we have
\begin{equation}
  \tP{n} = \del{x_n} - (n-1) \bigl( U_{n-1} \del{U_0} + U_{1\,n-1} \del{U_{01}} + U_{2\,n-1} \del{U_{02}} \bigr), 
  \quad 2 \leq n \leq \ell+\hf
  \label{PnExt7}
\end{equation}
For $ n \geq \ell + \th $ the generator $ \tP{n} $ is obtained by collecting $ t=0 $ terms of (\ref{PnExt6}). 
It has a slightly different form depending on the value of $n.$ 
For $ n = \ell + \th $ it is given by
\begin{equation}
  \tP{\ell+\th} = -\bigl( \ell+\hf \bigr) \Bigl[ 
      U_{\ell+\hf} \del{U_0} + \sum_{k=1,2} U_{k\,\ell+\hf} \del{U_{0k}} 
      + a_{\ell} \Bigl( U \del{U_{\ell+\hf}} +  \sum_{m=1}^{\ell+\hf} U_m \del{U_{m\, \ell+\hf}} + U_{\ell+\hf} \del{U_{\ell+\hf\,\ell+\hf}} \Bigr) \Bigr],
   \label{PnExt8}
\end{equation}
where
\[
 a_{\ell} = \Bigl( \bigl( \ell-\hf \bigr) ! \Bigr)^2.
\]
For other values of $n$ they are given by
\begin{eqnarray}
  \tP{n} &=& -I_{n-1} \Bigl[ -(2\ell+2-n) x_{2\ell+3-n} \Bigl( U \del{U_0} + \sum_{k=1,2}  U_k \del{U_{0k}} \Bigr) + U \del{U_{2\ell+2-n}}
   \nonumber \\
   &+& \sum_{k=1}^{\ell+\hf} U_k \del{U_{k\,2\ell+2-n}} + U_{ 2\ell+2-n} \del{U_{2\ell+2-n\,2\ell+2-n}} \Bigr],
   \qquad \ell+\frac{5}{2} \leq n \leq 2\ell-1
   \nonumber \\
  \tP{2\ell} &=& -I_{2\ell-1} \bigl[ -2x_3 (U \del{U_0} + U_1 \del{U_{01}}) + U \del{U_2} + (U_0 - 2 x_3 U_2) \del{U_{02}}
   + \sum_{k=1}^{\ell+\hf} U_k \del{U_{2k}} + U_2 \del{U_{22}} \bigr],
   \nonumber
\end{eqnarray}
and
\[
  \tP{2\ell+1} = -I_{2\ell} \bigl[ -x_2 U \del{U_0} + U \del{U_1} + (U_0 - x_2 U_1) \del{U_{01}} - (U+x_2 U_2) \del{U_{02}} 
   + \sum_{k=1}^{\ell+\hf} U_k \del{U_{1k}} + U_1 \del{U_{11}} \bigr].
\]
The best way to solve the symmetry condition is to start from the larger values of $n$. 
We first investigate the symmetry conditions for $ P^{(2\ell+1)} $ to $ P^{(\ell+\th)} $ in this order.  
They are separated in three cases (two cases for $ \ell=\frac{5}{2}$). 
\begin{Lemma} \label{Lem:LastTwo}
\begin{enumerate}
\renewcommand{\labelenumi}{(\roman{enumi})}
\item 
Equation (\ref{InvEqForm3}) is invariant under $ P^{(2\ell+1)} $ and $ P^{(2\ell)} $ if it has the form
  \begin{equation}
     F\left(x_a, \frac{U_0}{U}, \tilde{\phi}, \frac{U_k}{U}, \phi_{01}, \phi_{02}, \phi_{1b}, \phi_{2b}, U_{km} \right) = 0, \quad 
     a \in \I{2}, \  k, m \in \I{3}, \ b \in \I{1} 
      \label{InvEqForm4}
  \end{equation}
  where
  \begin{eqnarray}
    \tilde{\phi} &=& \frac{U_0}{U} + x_2 \frac{U_1}{U} + 2x_3 \frac{U_2}{U},
    \nonumber \\
    \phi_{01} &=& \frac{U_{01}}{U} - \frac{U_0 U_1}{U^2}, \qquad 
    \phi_{02} = \frac{U_{02}}{U} + \frac{U_1}{U} - \frac{U_0 U_2}{U^2},
    \nonumber \\
    \phi_{\alpha k} &=& \frac{U_{\alpha k}}{U} - \frac{U_{\alpha} U_k}{U^2}, \quad \alpha =1, 2
    \label{phi-def}
  \end{eqnarray}
\item 
Equation (\ref{InvEqForm4}) is invariant under $ P^{(n)}, \ \ell+\frac{5}{2} \leq n \leq 2\ell-1 $  if it has the form 
 \begin{equation}
    F\left(x_a, \phi, \frac{U_{\ell+\hf}}{U}, \phi_{01}, \phi_{02}, \phi_{km} \right)= 0, \quad 
    a \in \I{2}, \ k, m \in \I{1}
    \label{InvEqForm5} 
 \end{equation}
 where $ \phi_{01}, \phi_{02} $ are given in (\ref{phi-def}) and 
 \begin{eqnarray}
   \phi = \frac{U_0}{U} + \sum_{j=1}^{\ell-\hf} j x_{j+1} \frac{U_j}{U}, \qquad 
   \phi_{km} = \frac{U_{km}}{U} - \frac{ U_k U_m }{U^2}.
   \label{phi-def2}
 \end{eqnarray}
\item 
Equation (\ref{InvEqForm5}) is invariant under $ P^{(\ell+\th)} $  if it has the form 
 \begin{equation}
    F\bigl( x_a, w, w_{01}, w_{02}, \phi_{km}\bigr) = 0,  \quad a \in \I{2}, \ k,m \in \I{1}
    \label{InvEqForm6}  
 \end{equation}
 where $ \phi_{km} $ is given in (\ref{phi-def2}) and 
 \begin{equation}
   w = \phi - \frac{U_{\ell+\hf}^2}{2 a_{\ell} U^2}, \qquad 
   w_{0\alpha} = \phi_{0\alpha} - \frac{\phi_{\alpha\,\ell+\hf}}{a_{\ell}} \frac{U_{\ell+\hf}}{U}, \quad 
   \alpha = 1, 2
   \label{w-def}
 \end{equation}
 The constant $ a_{\ell} $ is defined below the equation (\ref{PnExt8}). 
\end{enumerate}
For $ \ell=\frac{5}{2}$ we have the cases (i) and (iii). 
\end{Lemma}
\begin{proof}[Proof of Lemma \ref{Lem:LastTwo}]
(i) 
 The symmetry conditions $ \tP{2\ell+1} F = \tP{2\ell} F =  0 $ is a system of first order PDEs. 
They are solved by the standard technique and it is not difficult to see that the $ \tilde{\phi} $ and  $\phi$'s 
given in (\ref{phi-def}) are solutions to the system of PDEs. 
\\
(ii) 
 It is immediate to verify that $\phi_{01}, \phi_{02}$ solve the symmetry conditions $ \tP{n} F = 0 $ for 
 $ \ell+\frac{5}{2} \leq n \leq 2\ell-1. $ Rewriting the symmetry conditions in terms of the variables given in (\ref{phi-def}) 
it is not difficult to solve them and find $ \phi $ and $ \phi_{km} $ in (\ref{phi-def2}) are the solutions. 
\\
(iii)  
 It is immediate to see that all $ \phi_{km}, \ k, m \in \I{1} $ solves the symmetry condition  $ \tP{\ell+\th} F = 0, $ 
however, $ \phi, \phi_{01} $ and $ \phi_{02} $ do not. Rewriting the symmetry condition in terms of $\phi$'s then 
solving the condition is an easy task. One may see that the variables in (\ref{w-def}) are solution of it. 
\end{proof}
\begin{Theorem} \label{Thm:HMPall}
The PDE invariant under the group generated by $ \h $ with $ \ell \geq \frac{5}{2} $ is given by
 \begin{equation}
    F\bigl( w, \alpha_n, \beta_n, \phi_{km} \bigr) =0, \quad n \in \I{2}, \ k,m \in \I{1}
    \label{InvEqForm7}      
 \end{equation}
 where  $ F$ is an arbitrary differentiable function and 
 \begin{eqnarray}
   \alpha_n &=& w_{01} + (n-1) x_n \phi_{1\,n-1}
   \nonumber\\
   &=& \frac{U_{01}}{U} - \frac{U_0 U_1}{U^2} 
    - \frac{U_{\ell+\hf}}{a_{\ell} U} \left( \frac{U_{1\,\ell+\hf}}{U} - \frac{U_1 U_{\ell+\hf}}{U^2} \right)
    + (n-1) x_n \left( \frac{U_{1\,n-1}}{U} - \frac{U_{1} U_{n-1}}{U^2} \right),
   \nonumber \\
   \beta_n &=& w_{02} + (n-1) x_n \phi_{2\,n-1}
   \nonumber\\
   &=& \frac{U_{02}}{U} + \frac{U_1}{U} - \frac{U_0 U_2}{U^2} 
    - \frac{U_{\ell+\hf}}{a_{\ell} U} \left( \frac{U_{2\,\ell+\hf}}{U} - \frac{U_2 U_{\ell+\hf}}{U^2} \right)   
    + (n-1) x_n \left( \frac{U_{2\,n-1}}{U} - \frac{U_{2} U_{n-1}}{U^2} \right).
    \label{ab-def}
 \end{eqnarray}
\end{Theorem}
\begin{proof}[Proof of Theorem \ref{Thm:HMPall}]
 Theorem is proved by making the equation (\ref{InvEqForm6}) invariant under $ P^{(n)} $ with $ 2 \leq n \leq \ell+\hf. $ 
 It is easy to see that $ w $ and all $ \phi_{km} $ solve the symmetry conditions $ \tP{n} F = 0 $ with $ \tP{n} $ given by (\ref{PnExt7}). 
 Thus the symmetry conditions are written in terms of only $ x_n $ and $ w_{0\alpha}: $ 
 \[
      \Big( \del{x_n} - (n-1) ( \phi_{1\,n-1} \del{w_{01}} + \phi_{2\,n-1} \del{w_{02}}) \Bigr) F = 0, \quad n \in \I{2}
 \]
 It is easily verified that the solutions of this system of equations are given by $ \alpha_n $ and $ \beta_n.$ 
 Thus we have proved the theorem.
\end{proof} 

%
\section{The case of $ \ell \geq \frac{5}{2} $ : $\g$-symmetry}
\label{Sec:g}

  Our next task is to make the equation (\ref{InvEqForm7}) invariant under $ D$ and $C. $ 
From the equation (\ref{DExt}) one may see that $ D $ generates the following scaling:
\begin{equation}
  w \ \to \ e^{-2\epsilon} w, \quad \alpha_n \ \to \ e^{-2(\ell+1)\epsilon} \alpha_n, \quad 
  \beta_n \ \to \ e^{-2\ell\epsilon} \beta_n, \quad 
  \phi_{km} \ \to \ e^{-2(2\ell+2-k-m)\epsilon} \phi_{km}.
  \label{Dscale2}
\end{equation}
Now we need the prolongation of $C$ up to second order.  
After lengthy but straightforward computation one may obtain the formula:
\begin{eqnarray}
 \hC &=& - \frac{ b_{\ell} }{2} x_{\ell+\hf}^2 \hM + t \hD +2\ell x_1 \tP{2} - \tC, 
 \nonumber \\
 \tC &=& - \sum_{k=2}^{\ell-\hf} \lambda_k x_k \del{x_{k+1}} + \sum_{k=2}^{\ell+\hf} 2(\ell+1-k) x_k U_k \, \del{U_0}
 \nonumber\\
 &+& \sum_{k=1}^{\ell-\hf} \Bigl[
     \lambda_k U_{k+1} \del{U_k} + \sum_{m=k}^{\ell-\hf} (\lambda_k U_{k+1\,m}+ \lambda_m U_{k\,m+1}) \del{U_{km}}
     + (\lambda_k U_{k+1\,\ell+\hf} + b_{\ell} x_{\ell+\hf} U_k) \del{U_{k\,\ell+\hf}}
   \Bigr]
 \nonumber \\
 &+& \sum_{k=1,2} \Bigl[
      2(\ell+1-k) U_k + \sum_{m=2}^{\ell+\hf} 2 (\ell+1-m) x_m U_{km} + \lambda_k U_{0\,k+1}
   \Bigr] \del{U_{0k}}
 \nonumber \\
 &+& 
   b_{\ell} \Bigl[
     x_{\ell+\hf} U\, \del{U_{\ell+\hf}} + (U + 2 x_{\ell+\hf} U_{\ell+\hf} ) \del{U_{\ell+\hf\,\ell+\hf}} \,
   \Bigr],
 \label{CExt2}
\end{eqnarray}
where
\[
 b_{\ell}= \Bigl( \bigl(\ell+\hf \bigr)! \Bigr)^2, \qquad \lambda_k = 2\ell+1-k.
\]
One may ignore $\hM $ and $ \tP{2}$ since we have already taken them into account. 
$ \tC $ is independent of $ t $ so that the invariance under $ D $ and $ C $ is reduced to the one under 
$ D $ and $ \tC. $ An immediate consequence of the equation (\ref{CExt2}) of $ \tC $ is that $F$ may not depend on 
$ \alpha_n $ and $ \beta_n:$
\begin{Lemma} \label{Lem:ab-indep}
A necessary condition for the invariance of the equation (\ref{InvEqForm7}) under $C $ is 
that the function $F$ is independent of $ \alpha_n $ and $ \beta_n.$ 
\end{Lemma}
\begin{proof}[Proof of Lemma \ref{Lem:ab-indep}] 
 $\tC$ has the term $ U_{03} \del{U_{02}} $ and this is the only term having $ U_{03}. $ 
On the other hand  $ F$ is independent of $ U_{03} $ so that we have the condition $ \del{U_{02}} F = 0. $ 
This means that $F$ is independent of $ U_{02}, $ i.e., independent of $ \beta_n. $  
$ \tC $ also has the term $ U_{02} \del{U_{01}} $ and this is the only term having $ U_{02}. $ 
Thus by the same argument $ F$ is not able to depend on $ U_{01}, $ i.e., $ \alpha_n. $ 
\end{proof}

 Now we turn to the variables $ w $ and $ \phi_{km}. $ 
It is immediate to  see that  $ w $ is  an invariant of $\tC, $ however, $\phi_{km}$'s are not:
\begin{eqnarray}
  \tC w &=& 0, \nonumber \\
  \tC \phi_{km} &=&\lambda_k \phi_{k+1\,m} + \lambda_m \phi_{k\,m+1}, 
  \qquad \tC \phi_{k\,\ell+\hf} = \lambda_k \phi_{k+1\,\ell+\hf}, 
  \quad 1 \leq k, m \leq \ell-\hf
  \nonumber \\
  \tC \phi_{\ell+\hf\,\ell+\hf} &=&  b_{\ell}.  \label{tConphi}
\end{eqnarray}
Thus the generator $\tC $ has the simpler form in terms of $ \phi_{km} $ 
(we omit $ \del{U_{0k}}$):
\begin{equation} 
  \tC = \sum_{k=1}^{\ell-\hf} \sum_{m=k}^{\ell-\hf} (\lambda_k \phi_{k+1\,m} + \lambda_m \phi_{k\,m+1}) \del{\phi_{km}}
  + \sum_{k=1}^{\ell-\hf} \lambda_k \phi_{k+1\,\ell+\hf} \,\del{\phi_{k\,\ell+\hf}}
  + b_{\ell}\, \del{\phi_{\ell+\hf\,\ell+\hf}}.
  \label{CExt3}
\end{equation}
The characteristic equation of the symmetry condition $ \tC F = 0 $ is a system of the first order PDEs given by
\begin{eqnarray}
 & & \frac{ d\phi_{km} }{\lambda_k \phi_{k+1\,m} + \lambda_m \phi_{k\,m+1}} = \frac{ d\phi_{\ell+\hf\,\ell+\hf} }{b_{\ell}},
     \qquad 1 \leq k \leq m \leq \ell-\hf 
     \label{CharEqC1}
     \\[3pt]
 & & \frac{ d\phi_{k\,\ell+\hf} }{ \lambda_k \phi_{k+1\,\ell+\hf} } = \frac{  d\phi_{\ell+\hf\,\ell+\hf} }{b_{\ell}},
     \qquad 1 \leq k \leq \ell-\hf
     \label{CharEqC2}
\end{eqnarray}
One may solve it recursively by starting with (\ref{CharEqC2}) for $ k = \ell-\hf:$
\[
   \frac{ d\phi_{\ell-\hf\,\ell+\hf} }{d \phi} = \frac{\lambda_{\ell-\hf}}{b_{\ell}} \phi, 
   \qquad 
   \phi = \phi_{\ell+\hf\,\ell+\hf}. 
\]
This gives the invariant of $ \tC: $
\begin{equation}
  w_{\ell-\hf\,\ell+\hf} = \phi_{\ell-\hf\,\ell+\hf} - \frac{ \lambda_{\ell-\hf} }{ b_{\ell} } \frac{\phi^2}{2}.
  \label{w-first}
\end{equation}
Next we rewrite the equation (\ref{CharEqC2}) for $ k = \ell-\th $ in the following way:
\[
   \frac{ d\phi_{\ell-\th\,\ell+\hf} }{d \phi} = \frac{ \lambda_{\ell-\th} }{ b_{\ell} } \phi_{\ell-\hf\,\ell+\hf}
   = \frac{ \lambda_{\ell-\th} }{ b_{\ell} } 
   \left(
     w_{\ell-\hf\,\ell+\hf} + \frac{ \lambda_{\ell-\hf} }{ b_{\ell} } \frac{\phi^2}{2}
   \right).
\]
Then we find an another invariant:
\begin{equation}
  w_{\ell-\th\,\ell+\hf} = \phi_{\ell-\th\,\ell+\hf} 
  - \frac{ \lambda_{\ell-\th} }{ b_{\ell} } w_{\ell-\hf\,\ell+\hf}\, \phi
  - \frac{ \lambda_{\ell-\th} \lambda_{\ell-\hf} }{ b_{\ell}^2 } \frac{\phi^3}{3!}.
  \label{w-second}
\end{equation}
The complete list of invariants of $ \tC $ is given as follows:
\begin{Lemma} \label{Lem:wkm}
 Solutions of the  equations (\ref{CharEqC1}) and (\ref{CharEqC2}) are given by 
\begin{eqnarray}
 && 
  w_{km} = \phi_{km} - \sum_{a+b \geq 1} c_{ab}(k,m) w_{k+a\,m+b} \frac{ \phi_{\ell+\hf\,\ell+\hf}^{a+b} }{(a+b)!}
  -\gamma(k,m) \frac{ \phi_{\ell+\hf\,\ell+\hf}^{2\ell+2-k-m} }{(2\ell+2-k-m)!},
  \label{Cinvariants-km}
 \\ 
 &&  
  \hspace{7.5cm} 1 \leq k \leq \ell-\hf, \ k \leq m \leq \ell+\hf \nonumber
\end{eqnarray}
where $ a, b $ run over nonnegative integers such that  $ a \leq \ell-\hf-k, \ b \leq \ell + \hf -m $ 
and $ k + a \leq m+b. $ The coefficient $ \gamma(k,m) $ depends on $ c_{ab}(k,m) $ with the maximal value of $a $ and $b:$
\begin{equation}
  \gamma(k,m) = 
    \left\{
      \begin{array}{lcl}
       \displaystyle \frac{ \lambda_{\ell-\hf} }{b_{\ell}}, & & (k,m) = (\ell-\hf,\ell+\hf) \\[15pt]
       \displaystyle \frac{ \lambda_{\ell-\hf} }{b_{\ell}}\, c_{\max(a)\,\max(b)}(k,m), & & \text{otherwise}
      \end{array}
     \right.
  \label{gamma-values}
\end{equation} 
The coefficients $ c_{ab}(k,m) $ are calculated by the algorithm given below. 
\end{Lemma}
\noindent
\textbf{Algorithm.} 
We borrow the terminology of graph theory.
\begin{enumerate}
\renewcommand{\labelenumi}{(\arabic{enumi})}
\item For a given $ w_{km},$ draw a rooted tree  according to the branching rules given in Figure \ref{Fig:one}. 
Each vetex and each edge of this tree are labelled. The root is labelled by $ w_{km}.$  
Other vertices and edges are labelled as indicaed in Figure \ref{Fig:one}. 
Each vertex has at most two children according to its label.  
The vertex has no children if its label is  $ w_{\ell-\hf\,\ell+\hf}. $ Thus the hight of the tree is $ 2\ell-k-m. $  
An example for $ \ell = \frac{7}{2} $ is indicated in Figure \ref{Fig:two}. 
\item Take a directed path from the root to one of the verticies with label $ w_{k+a,m+b} $ and multiply all the edge labels on this path. For instance, take the path $ (w_{13}, w_{14}, w_{24} ) $ in Figure \ref{Fig:two}. 
Then the multiplication of the labels is $ \lambda_1 \lambda_3 b_{7/2}^{-2}. $
\item If there exit other vertices whose label is also $ w_{k+a,m+b} $ (same label as (2)), 
then  repeat the same computation as (2) for  the direct paths to such vertices. In Figure \ref{Fig:two} there is one more vertex whose label is $ w_{24} $ and the path is $ (w_{13},w_{23}, w_{24}). $  
 We have $ \lambda_1 \lambda_3 b_{7/2}^{-2} $ for this path, too.
\item  Take summation of all such multiplication for the paths to the vertices whose label is $ w_{k+a\,m+b}, $  then this summation gives the coefficient $ c_{ab}(k,m).$ 
For the tree in Figure \ref{Fig:two} the coefficient of $ w_{24} $ is obtained by adding the quantities calculated in (2) and (3):  
$ c_{11}(1,3) = 2 \lambda_1 \lambda_3 b_{7/2}^{-2}. $
\end{enumerate}
\begin{figure}[h]
  \begin{center}
    \includegraphics[width=110mm]{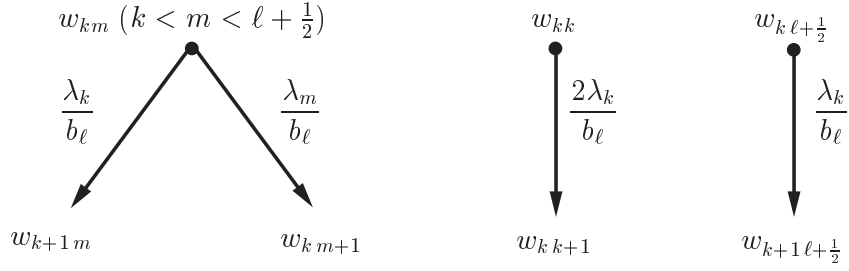}
  \end{center}
  \caption{Vertices and edges}
  \label{Fig:one}
\end{figure}
\begin{figure}[h]
  \begin{center}
    \includegraphics[keepaspectratio=true,height=95mm]{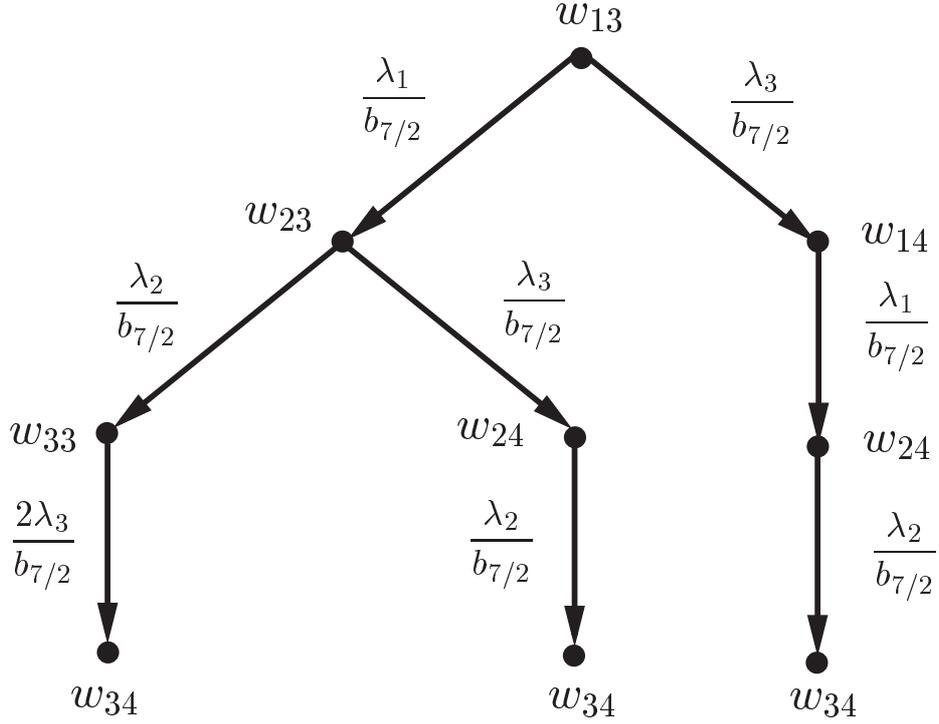}
  \end{center}
  \caption{Example of rooted tree: $ \ell=\frac{7}{2}$}
  \label{Fig:two}
\end{figure}


\begin{proof}[Proof of Lemma \ref{Lem:wkm}] 
The lemma is proved by induction on height of the trees. 
We have a tree of height zero only when label of the root is $ w_{\ell-\hf\,\ell+\hf}. $  
In this case no $ c_{ab} $ appears so that (\ref{Cinvariants-km}) yields
\[
  w_{\ell-\hf\,\ell+\hf} = \phi_{\ell-\hf\,\ell+\hf} - \gamma\Bigl(\ell-\hf,\ell+\hf \Bigr) \frac{\phi^2_{\ell+\hf\,\ell+\hf}}{2}.
\]
This coincide with (\ref{w-first}). 
To verify the legitimacy of the algorithm calculating $ c_{ab}(k,m) $ we need to start with a tree of height one. 
There are two possible labels of the root to obtain a tree of height one. 
They are $ w_{\ell-\th\,\ell+\hf} $ and $ w_{\ell-\hf\,\ell-\hf}. $
Let us start with  the label  $ w_{\ell-\th\,\ell+\hf}. $ 
It is not difficult to verify, by employing the algorithm,  that we obtain the equation (\ref{w-second}) for this case. 
For the  label $ w_{\ell-\hf\,\ell-\hf}, $ the algorithm gives the following result:
\begin{equation}
 w_{\ell-\hf\,\ell-\hf}  = \phi_{\ell-\hf\,\ell-\hf} - \frac{ 2\lambda_{\ell-\hf} }{b_{\ell}} w_{\ell-\hf\,\ell+\hf} \phi_{\ell+\hf\,\ell+\hf}
 - 2\left(  \frac{ 2\lambda_{\ell-\hf} }{b_{\ell}} \right)^2 \frac{\phi^3_{\ell+\hf\,\ell+\hf}}{3!}.
 \label{w-third}
\end{equation}
It is easy to see that $ \tC $ annihilates (\ref{w-third}). Thus the lemma is  true for trees of height one.  

 Now we consider trees of height $ h > 1.$  
If  label of the root is $ w_{km} \ (k < m < \ell+\hf), $ then the tree has two rooted subtrees (height $ h-1$) 
such that one of then has the root whose label is $ w_{k+1\,m} $ and another has the root whose label is $ w_{k\,m+1}. $  
On the other hand, if label of the root is  $ w_{kk} $ or $ w_{k\,\ell+\hf}, $ then the tree has only one rooted subtree
 (height $ h-1$) such that the subtree has the root whose label  is  $ w_{k\,k+1} $ or $ w_{k+1\,\ell+\hf}. $ 
By the algorithm one may find relations between the coefficients $ c_{ab}, \ \gamma $ for the tree of height $h$ and  
the subtrees of height $ h-1:$
\begin{eqnarray}
  c_{ab}(k,m) &=& \frac{\lambda_k}{b_{\ell}} c_{a-1\,b}(k+1,m) + \frac{\lambda_m}{b_{\ell}} c_{a\,b-1}(k,m+1),
  \nonumber \\
  \gamma(k,m) &=& \frac{\lambda_k}{b_{\ell}} \gamma(k+1,m) + \frac{\lambda_m}{b_{\ell}} \gamma(k,m+1),
  \nonumber \\
  c_{ab}(k,k) &=& \frac{2\lambda_k}{b_{\ell}} c_{a\,b-1}(k,k+1), \qquad 
  \gamma(k,k) = \frac{2\lambda_k}{b_{\ell}} \gamma(k,k+1).
  \label{c-gamma-rels}
\end{eqnarray}
We understand that $ c_{ab} $ and $ \gamma $ are zero if their indices or arguments have a impossible value. 

  Assumption of the induction is that the lemma is true for any rooted subtrees whose height is smaller than $h.$ 
Namely, we assume that $ \tC w_{k+a\,m+b} = 0 $ for  $ a+b \geq 1 $  
and what we need to show is that $ \tC w_{km} = 0. $ 
We separate out $ a+b = 1 $ terms  from the summation in  (\ref{Cinvariants-km}) and use (\ref{tConphi}) to calculate 
the action of $ \tC $ on $ w_{km}. $ For $ k < m < \ell+\hf$ we have
\begin{eqnarray*}
  \tC w_{km} &=& \tC \phi_{km} - \lambda_k w_{k+1\,m} - \lambda_m w_{k\,m+1} 
  \\
  &-& b_{\ell} \sum_{a+b\geq 2} c_{ab}(k,m) w_{k+a\,m+b} \frac{ \phi_{\ell+\hf\,\ell+\hf}^{a+b-1} }{(a+b-1)!}
  - b_{\ell} \gamma(k,m) \frac{ \phi_{\ell+\hf\,\ell+\hf}^{2\ell+1-k-m} }{(2\ell+1-k-m)!}
  \\
  &=& \tC \phi_{km} - \lambda_k w_{k+1\,m} - \lambda_m w_{k\,m+1} 
  \\
  &-& \sum_{a+b \geq 1} \bigl( \lambda_k c_{ab}(k+1,m) w_{k+1+a\,m+b} + \lambda_m c_{ab}(k,m+1) w_{k\,m+1} \bigr) 
  \frac{ \phi_{\ell+\hf\,\ell+\hf}^{a+b} }{(a+b)!}
  \\
  &-& \bigl( \lambda_k \gamma(k+1,m) + \lambda_m \gamma(k,m+1) \bigr) 
  \frac{ \phi_{\ell+\hf\,\ell+\hf}^{2\ell+1-k-m} }{(2\ell+1-k-m)!}
\end{eqnarray*}
The second equality is due to the relations (\ref{c-gamma-rels}) and the replacement $ a-1 $ (resp. $b-1$) with $a $ (resp. $b$). 
By the assumption of the induction one may use (\ref{Cinvariants-km}) to obtain:
\[
  \tC w_{km} = \tC \phi_{km} - (\lambda_k \phi_{k+1\,m} + \lambda_m \phi_{k\,m+1}) = 0.
\]
The second equality is due to (\ref{tConphi}). 

  The proof of $ \tC w_{kk} = \tC w_{k\,\ell+\hf} = 0 $ is done in a similar way. This completes the proof of Lemma \ref{Lem:wkm}.
\end{proof}

\begin{Corollary} \label{Cor:w}
 The variables $ w_{k\,\ell+\hf} \ (1 \leq k \leq \ell-\hf)$ are easily calculated by this method.
 \[
     w_{k\,\ell+\hf}  = \phi_{k\,\ell+\hf} - \sum_{n=1}^{\ell-\hf-k} 
     \begin{pmatrix}  2\ell+1-k \\ n \end{pmatrix} w_{k+n\,\ell+\hf} 
     \left( \frac{  \phi_{\ell+\hf\,\ell+\hf}  }{ b_{\ell} } \right)^n
 - \begin{pmatrix}  2\ell+1-k \\ \ell+\hf \end{pmatrix} 
   \frac{1}{ b_{\ell}^{\ell+\hf-k} } 
   \frac{  \phi_{\ell+\hf\,\ell+\hf}^{\ell+\th-k}}{\ell+\th-k}.
 \]
\end{Corollary}

\begin{proof}[Proof of Corollary \ref{Cor:w}]
  The rooted tree used for this computation is indicated in Figure \ref{fig:figure3.eps}. 
It follows that the coefficient of $ w_{k+a\,\ell+\hf} $ is given by
\[
  c_{a0}(k,{\scriptstyle \ell+\hf}) = \frac{ \lambda_k \lambda_{k+1} \cdots \lambda_{k+a-1}  }{ b_{\ell}^a  }
  = \begin{pmatrix}
       2\ell + 1 - k \\ a
  \end{pmatrix} 
  \frac{a!}{ b_{\ell}^a }.
\]
By (\ref{gamma-values}) the  coefficient $ \gamma(k,\ell+\hf) $ is calculated as 
$  \gamma(k,\ell+\hf) = \lambda_{\ell-\hf} b_{\ell}^{-1} c_{\ell-\hf-k\,0}(k,\ell+\hf). $ 
Thus we obtain  the expression of $ w_{k\,\ell+\hf} $ given in the corollary.
\end{proof}
\begin{figure}[h]
  \begin{center}
    \includegraphics[keepaspectratio=true,height=17mm]{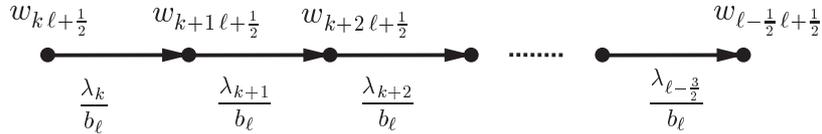}
  \end{center}
  \caption{Rooted tree for the computation of $w_{k\,\ell+\hf}$}  \label{fig:figure3.eps}
\end{figure}

 Our final task is to consider the invariance under $D.$ It is immediate to see that $ \hD $ scales $ w_{km} $ as 
\[
  w_{km} \ \to \ e^{-2(2\ell+2-k-m) \epsilon} \,w_{km}.
\]
Together with the scaling law (\ref{Dscale2}) we arrive at the final theorem. 
\begin{Theorem} \label{Thm:Main}
The PDE invariant under the group generated by $ \g $ with $ \ell \geq \frac{5}{2} $ is given by
 \begin{equation}
    F\left( \frac{w_{km}}{ w^{2\ell+2-k-m} } \right) =0, \quad 1 \leq k \leq \ell-\hf, \ k \leq m \leq \ell+\hf
    \label{InvEqThm2}      
 \end{equation}
 where  $ F$ is an arbitrary differentiable function. The variables $ w $ and $ w_{km}$  are given in (\ref{w-def}) and 
(\ref{Cinvariants-km}), respectively.  
This is the PDE with $ \ell+\th $ independent and one dependent variables. The function $F$ has 
$ \hf \Big( \ell-\hf \Big) \Big( \ell+\frac{5}{2}\Big) $ arguments. 
\end{Theorem}
\begin{Example} 
Invariant PDE for $ \ell = \frac{5}{2}.$ 
\[
   F\left( \frac{w_{11}}{w^5}, \frac{w_{12}}{w^4}, \frac{w_{13}}{w^3}, \frac{w_{22}}{w^3}, \frac{w_{23}}{w^2} \right) = 0 
\]
where
\begin{eqnarray*}
  w &=&  \frac{U_0}{U} + x_2 \frac{U_1}{U} + 2 x_3 \frac{U_2}{U} - \frac{U_3^2}{8U^2},
  \\
  w_{23} &=& \phi_{23} - \frac{1}{18} \phi_{33}^2,
  \\
  w_{22} &=& \phi_{22} - \frac{2}{9} \phi_{23} \phi_{33} + \frac{2}{3^5} \phi_{33}^3,
  \\
  w_{13} &=& \phi_{13}- \frac{5}{36} \phi_{23} \phi_{33} + \frac{5}{2^23^5} \phi_{33}^3,
  \\
  w_{12} &=& \phi_{12} - \frac{1}{9} \phi_{13} \phi_{33} - \frac{5}{36} \phi_{22} \phi_{33}  +\frac{5}{2^3 3^3} \phi_{23} \phi_{33}^2 
            - \frac{5}{2^5 3^5}  \phi_{33}^4,
  \\
 w_{11} &=& \phi_{11} - \frac{5}{18} \phi_{12} \phi_{33} +\frac{5}{2^2 3^4} \phi_{13} \phi_{33}^2  - \frac{25}{2^4 3^4} \phi_{22} \phi_{33}^2  
            - \frac{25}{2^4 3^6} \phi_{23} \phi_{33}^3   - \frac{5}{2^4 3^8}  \phi_{33}^5.
 \end{eqnarray*}
 The symmetry generators are given by
 \begin{eqnarray}
    & & M = U \del{U}, \qquad \,D = 2t \del{t} + 5 x_1 \del{x_1} + 3x_2 \del{x_2} + x_3 \del{x_3}, \qquad 
        H = \del{t},   
    \nonumber \\
    & & C = t(t\del{t} + 5 x_1 \del{x_1} + 3 x_2 \del{x_2} + x_3 \del{x_3})
          + 5 x_1 \del{x_2} + 4 x_2 \del{x_3} - 18 x_3^2 U \del{U},
    \nonumber \\
    & & P^{(1)} = \del{x_1}, \qquad P^{(2)} = t \del{x_1} + \del{x_2}, \qquad
    P^{(3)} = t^2 \del{x_1} + 2t \del{x_2} + \del{x_3},
    \nonumber \\
    & & P^{(4)} = t^3 \del{x_1} +  3t^2 \del{x_2} + 3t \del{x_3} -12 x_3 U \del{U},
    \nonumber \\
    & & P^{(5)} = t^4 \del{x_1} + 4t^3 \del{x_2} + 4t^2 \del{x_3}-24(2tx_3 + x_2) U \del{U},
    \nonumber \\
    & & P^{(6)} = t^5 \del{x_1} + 5t^4 \del{x_2} + 10t^3 \del{x_3} -120(t^2 x_3 - t x_2 + x_1) U\del{U}.\nonumber
 \end{eqnarray}
\end{Example}


\section{Concluding remarks}
\label{Sec:CR}

 We have constructed nonlinear PDEs invariant under the transformations generated by 
the realization of CGA given in (\ref{CGAgenerators}). 
This was done by obtaining the general solution of the symmetry conditions so that the PDEs constructed in this work 
are the most general ones invariant under (\ref{CGAgenerators}). 
A remarkable property of the PDEs is that they do not contain the second order derivative in $t $ if $ \ell > \frac{3}{2}. $ 
It means that there exist no invariant PDEs of wave or Klein-Gordon type for $ \ell > \th. $ 
This type of $\ell$-dependence does not appear in the linear PDEs constructed in \cite{Aizawa:2013vma, Aizawa:2014hva} based 
on the representation theory of $\g.$ 
This  will be changed if one start with a realization of CGA which is different from (\ref{CGAgenerators}). 

The CGAs considered in this work are only $d = 1$ members. 
Extending the present computation to higher values of $d$ would be an interesting future work. 
Because the $ d = 2 $ CGA has a distinct central extension so that we will have different types of invariant PDEs. 
For $ d \geq 3 $ CGAs have $ so(d)$ as a subalgebra. This will also cause a significant change in invariant PDEs.


\section*{Acknowledgments}

The authors are grateful to Prof. Y. Uno for helpful discussion. 
NA is supported by the  grants-in-aid from JSPS (Contract No.26400209).




\bibliographystyle{aip}
\bibliography{CGA}


%


%

\end{document}